\documentclass{amsart}
\usepackage{graphicx}
\usepackage{graphicx}
\usepackage{amssymb}
\usepackage{amsfonts}
\usepackage{amsmath}
\usepackage{amsthm}
\newtheorem{thm}{Theorem}[section]

\newtheorem{lem}[thm]{Lemma}
\newtheorem{rem}[thm]{Remark}
\newtheorem{cor}[thm]{Corollary}
\newtheorem{prop}[thm]{Proposition}

\numberwithin{equation}{section}

\newcommand{\ve}{{\varepsilon}}

\begin{document}
\title{Limit Theorems for Partial Hedging Under Transaction Costs}
 \author{Yan Dolinsky\\
 Department of Mathematics\\
 ETH, Zurich\\
 Switzerland }%

\address{
 Department of Mathematics, ETH, Zurich 8092, Switzerland\\
 {e.mail: yan.dolinsky@math.ethz.ch}}

 \date{\today}
\begin{abstract}
We study shortfall risk minimization for American options with path
dependent payoffs
under proportional transaction costs in the Black--Scholes (BS) model.
We show that for this case
the shortfall risk is a limit of similar terms in an appropriate sequence of binomial models.
We also
prove that in the continuous time BS model for a given initial capital
there exists a portfolio strategy which minimizes the shortfall risk.
In the absence of transactions costs (complete markets)
similar limit theorems were obtained in Dolinsky and Kifer (2008, 2010)
for game options. In the presence of transaction costs the markets are no
longer complete and additional machinery required. Shortfall risk minimization
for American options under transaction costs was not studied before.
\end{abstract}

\maketitle
\markboth{Y.Dolinsly}{Limit Theorems Under Transaction Costs}
\renewcommand{\theequation}{\arabic{section}.\arabic{equation}}
\pagenumbering{arabic}
\section{Introduction}\label{sec:1}\setcounter{equation}{0}
This paper deals with shortfall risk minimization for American options
under proportional transaction costs. It is well known
that in a complete market an American contingent claim can
be hedged perfectly with an initial capital which is equal to the
 optimal stopping value of the discounted payoff under the unique
martingale measure. In the presence of transaction costs the market
is no longer complete and the initial capital required for perfect
hedging (superhedging price) of the options is often too high.
In fact, several authors, see for example, Soner, Shreve and
Cvitanic (1995), Levental and Skorohod (1997) and Cvitanic, Pham and Touzi (1999)
showed that the superhedging price of European call
options (also of American call options) in the BS model is
equal to the price of buying the stock at the time the option
is purchased.
In Jakubenas, Levental, and Ryznar (2003)
these results were extended to path dependent options.
For example, it was demonstrated that
for European and American options (in the BS model) with Russian type of payoffs
the superhedging price
is infinite, i.e., perfect hedging is not available.
Thus with the presence of transaction costs
it is reasonable to assume that the seller's (investor's) initial
capital is less than the superhedging price.
In this case, the seller is
ready to accept a risk that his portfolio value at an exercise time
may be less than his obligation
to pay and he will need additional
funds to fullfil the contract. This leads to the natural question
of minimization of risk for a given amount of initial capital.
In order to make this question precise we need to define explicitly the risk measure.

We deal with a certain type of risk called the
shortfall risk, which is defined for American options
as the maximal expectation
with respect to the buyer exercise times of the discounted shortfall
(see Mulinacci 2010).
In the presence of transaction costs
the problem of shortfall risk minimization
was studied only for European options,
see Guasoni (2002A, 2002B), Komizono (2001, 2003)
and
Trivellato (2009). The first two authors considered a
general setup for which they proved that for
a given initial capital there exists a portfolio
strategy which minimizes the shortfall risk. In
Trivellato (2009) shortfall risk
minimization is studied for European options in a binomial
model and it is shown that for a given initial capital,
the shortfall risk and the corresponding optimal portfolio
can be calculated by dynamical programming algorithm.

In this paper we study shortfall risk
minimization for a cash--settled American options in the BS model.
We consider path dependent payoffs with some regularity conditions.
We
allow only self financing portfolios which satisfy
the no--bankruptcy condition i.e., a portfolios
with nonnegative wealth process.
This corresponds to the situation when the
portfolio is handled without borrowing of the capital.
By using convexity of the shortfall risk
measure, we will show that for a given initial capital
there exists a portfolio strategy which minimizes the risk.
From practical view point, existence results are not sufficient,
an investor with a fixed initial capital
want to compute the minimal possible shortfall risk
and to find explicitly a portfolio strategy which
minimizes or "almost"
minimizes the shortfall risk.
For binomial models
the above problems can be solved by
dynamical programming algorithm.
Our approach is to use an appropriate
sequence of binomial models in order
to approximate the shortfall risk
and to construct "almost" optimal portfolios
in the BS model. Namely, we will show that under proportional
transaction costs the shortfall risk in the BS model
is a limit of similar terms with the
same proportional transaction costs in
an appropriate sequence of
binomial models. Furthermore we will use the optimal portfolios in
the binomial models in order to construct "almost" optimal portfolios
for the BS model.

Similar results were obtained
in Dolinsky and Kifer (2008, 2010)
for game options
without the presence of
transaction costs. The proof of the results there relied
heavily on the completeness of the markets, which is
no longer the case with the presence
of transaction costs.

The main auxiliary result which is crucial for proving the limit
theorems in our setup is the stability of the shortfall as a function of the transaction costs parameters $\lambda,\mu$.
This result may be also of some independent interest. In particular
we will see that as $\mu,\lambda\downarrow 0$,
the shortfall risks converge to the shortfall risk of the complete market.
Note that for the superhedging prices this is not true
in general.
For instance, the call option superhedging prices converge (as
$\mu,\lambda\downarrow 0$) to the initial stock price which is
bigger than the call option price in the complete BS market. The same occurs for American
options with Russian type of payoffs. In this case the
limit of the superhedging prices (as $\mu,\lambda\downarrow 0$) is infinity.

The paper is organized as following. Main results of this paper are formulated in the next section.
In Section 3 we analyze the binomial models and
provide a dynamical programming algorithm for the shortfall risk
and the corresponding optimal portfolios.
In Section 4 we complete the proof of the limit theorems (Theorems 2.2--2.3).
In Section 5 we prove Theorem 2.1
which provide an existence result for the optimal portfolio in the BS model.

\section{Preliminaries and main results}\label{sec:2}\setcounter{equation}{0}
Consider a complete probability space
($\Omega_W$, $P^{W}$) together with a standard
one--dimensional continuous in time Brownian motion
\{$W(t)\}_{t=0}^\infty$, and the filtration
$\mathcal{F}^{W}_t=\sigma{\{W(s)|s\leq{t}\}}$.
We assume that the $\sigma$--algebras
contain the null sets.
A BS financial
market consists of a savings account $B(t)$
with an interest rate $r$, assuming without loss of generality
that $r=0$, i.e.
\begin{equation}\label{2.1}
B(t)\equiv B_0>0
\end{equation}
and of a risky asset $S^W$ given by the following equation
\begin{equation}\label{2.2}
S^W(t)=S_0\exp(\sigma W(t)+(\kappa-\sigma^2/2) t), \ \ S_0>0
\end{equation}
where $\sigma>0$ is called volatility and $\kappa\in\mathbb{R}$ is
another constant.
Denote by $\tilde{P}^W$
the unique martingale measure for the above
model. Using standard arguments it follows
that the restriction of
the probability measure $\tilde{P}^{W}$ to
the $\sigma$--algebra
$\mathcal{F}^{W}_t$ satisfies
\begin{equation}\label{2.2+}
Z(t):=\frac{d\tilde{P}^W}{dP^{W}}|\mathcal{F}^{W}_t
=\exp(-\frac{\kappa}{\sigma} W(t)-\frac{1}{2}(\frac{\kappa}{\sigma})^2 t).
\end{equation}

Let $T<\infty$ be the
maturity date of our American option and let $\mathcal{T}^W_{[0,T]}$ be the set of all stopping
times with respect to the filtration
$\mathcal{F}^W$ which take values in $[0,T]$. Denote by $M[0,T]$ the
space of all Borel measurable functions on $[0,T]$ with the
uniform topology (induced by the norm $||u||=\sup_{0\leq t\leq T}|u(t)|$).
Let $F:[0,T]\times M[0,T]\rightarrow\mathbb{R}_{+}$
be a continuous function (with respect to the product topology)
such that there exists a constant $C>0$ which satisfies
\begin{equation}\label{2.3}
\begin{split}
\sup_{0\leq t\leq T} F(t,x)\leq C\sup_{0\leq t \leq T} |x(t)|,
\ \ \forall{x}\in M[0,T].
\end{split}
\end{equation}
Furthermore, we assume that
for any $t\in [0,T]$ and $x,y\in M[0,T]$,
$F(t,x)=F(t,y)$ if $x(s)=y(s)$ for any $s\leq{t}$.

Next, consider a cash--settled American contingent claim with the
payoff process given by
\begin{equation}\label{2.4}
Y^W(t)=F(t,S^W), \ 0\leq t\leq{T}.
\end{equation}
From the assumptions above it follows that
${\{Y^W(t)\}}_{t=0}^T$ is a continuous
adapted stochastic process and
$E^W\sup_{0\leq t\leq T}Y^W(t),\tilde{E}^W\sup_{0\leq t\leq T}Y^W(t)<\infty$, where $E^W$ and
$\tilde{E}^W$, denote the expectations with respect to the probability
measures $P^W$ and $\tilde{P}^W$, respectively.

In our model purchase and sale, of
the risky asset are subject to a proportional
transaction costs of rate $\lambda$ and
$\mu$, respectively. We assume that $\lambda>0$ and $0<\mu<1$
are constants.
Thus a trading strategy with a (finite)
horizon $T$ and an initial capital $x$ is a pair
$\pi=(x,\gamma)$ where $\gamma=\{\gamma(t)\}_{t=0}^T$
is an adapted process of bounded variation with left continuous paths and $\gamma(0)=0$.
Set
\begin{equation}\label{2.5}
\gamma^{+}(t)=\frac{\gamma(t)+ \int_{0}^t |d\gamma(s)|}{2} \ \ \mbox{and} \ \
\gamma^{-}(t)=\frac {\int_{0}^t |d\gamma(s)|-\gamma(t)}{2}.
\end{equation}
Clearly $\gamma(t)=\gamma^{+}(t)-\gamma^{-}(t)$
is a decomposition of $\gamma$ into a positive variation $\gamma^{+}$
and a negative variation $\gamma^{-}$.
The random variables
$\gamma^{+}(t)$ and $\gamma^{-}(t)$,
denote the cumulative number of stocks,
purchased up to time $t$ and sold up to time $t$,
(not including the transfers made at time $t$) respectively.
The portfolio value at time $t\in [0,T]$ (after liqudation) of a trading strategy $\pi$ is given by
\begin{eqnarray}\label{2.6}
&V^{\pi}_{\lambda,\mu}(t)=x-(1+\lambda)
(\int_{0}^tS^W(u)d\gamma^{+}(u)+\gamma(t)^{-}S^W(t))\\
&+(1-\mu)
(\int_{0}^t S^W(u)d\gamma^{-}(u)+\gamma(t)^{+}S^W(t))\nonumber
\end{eqnarray}
where we denote $y^{+}=\max(y,0)$, $y^{-}=\max(-y,0)$.
Observe that $V^\pi_{\lambda,\mu}(t)$ is the portfolio
value before the transfers made at time $t$.
A self financing
strategy $\pi$ is called \textit{admissible}
if the following no--bankruptcy condition holds
\begin{equation}\label{2.7}
{V}^{\pi}_{\lambda,\mu}(t)\geq{0} \ \ \ \ \ \forall t\in{[0,{T}]}.
\end{equation}
The set of all \textit{admissible} self financing strategies with an initial
capital $x$ will denoted by $\mathcal{A}^W(x,\lambda,\mu)$.
For an \textit{admissible} self financing strategy
$\pi$ the shortfall risk is given by
\begin{equation}\label{2.8}
R(\pi,\lambda,u)=\sup_{\tau\in\mathcal{T}^W_{[0,T]}}E^W[(Y^W({\tau})-
{V}^{\pi}_{\lambda,\mu}(\tau))^+],
\end{equation}
which is the maximal possible expectation of the shortfall which measured in cash.
The shortfall risk for
an initial capital $x$ is given by
\begin{equation}\label{2.9}
R(x,\lambda,\mu)=\inf_{\pi\in\mathcal{A}^W(x,\lambda,\mu)} R(\pi,\lambda,\mu).
\end{equation}

A portfolio strategy
$\pi\in\mathcal{A}^W(x,\lambda,\mu)$ will be
called $\ve$-optimal if $R(\pi,\lambda,\mu)\leq R(x,\lambda,\mu)+\ve.$
For $\ve=0$ the above portfolio is called an optimal portfolio.

The following theorem (which is proved in Section 5) provides an existence result for the optimal portfolio.
\begin{thm}\label{thm2.0}
For any $\lambda>0$, $0<\mu<1$ and $x\in\mathbb{R}_{+}$, there exists a portfolio strategy
$\pi\in\mathcal{A}^W(x,\lambda,\mu)$ such that
\begin{equation}\label{2.9+}
R(\pi,\lambda,\mu)=R(x,\lambda,\mu).
\end{equation}
\end{thm}
Next, we introduce the binomial models.
Similar binomial models were used to approximate
option prices and shortfall risks
in the complete setup (see Kifer 2006, Dolinsky and Kifer 2008, 2010)
i.e., in the absence of transaction costs.
For any $n$ consider the $n$--step binomial market which consists
of a savings account $B^{(n)}(t)$ given by
\begin{equation}\label{2.10}
B^{(n)}(t)\equiv B_0>0
\end{equation}
and of risky stock $S^{\xi,n}$ given by
by the formulas
$S^{\xi,n}(t)=S_0$ for $t\in{[0,T/n)}$ and
\begin{equation}\label{2.11}
S^{\xi,n}(t)=S_0\exp\big(\sigma(T/n)^{1/2}\sum_{k=1}^{[nt/T]}
\xi_k\big)\,\,\mbox{if}\,\, t\geq{T/n}
\end{equation}
where $\xi_1,\xi_2,...$ are i.i.d. random variables taking values 1
and $-1$ with probabilities
$p^{(n)}=\big(\exp((\sigma-\frac{2\kappa}{\sigma})\sqrt{\frac{T}{n}})+1\big)^{-1}$
and
$1-p^{(n)}=\big(\exp((\frac{2\kappa}{\sigma}-\sigma)\sqrt{\frac{T}{n}})+1\big)^{-1}$,
respectively. Let $P^\xi_n={\{p^{(n)},1-p^{(n)}\}}^\infty$
be the corresponding product probability measure on the space of
sequences $\Omega_\xi=\{-1,1\}^\infty$. For any $k\geq 0$ let
$\mathcal{F}^{\xi}_k=\sigma{\{\xi_1,...,\xi_k}\}$,
($\mathcal{F}^\xi_0=\{\emptyset,\Omega_\xi\}$).
Denote
by $\mathcal{T}^{\xi}_{0,n}$ the set of all stopping times with
respect to the  filtration $\mathcal{F}^{\xi}_k$ with values in
${\{0,1,...,n\}}$.

The $n$--step binomial market
is active at the times $0,\frac{T}{n},\frac{2T}{n},...,T$.
As before we assume that purchase (respectively, sale)
of the risky asset is subject to a proportional
transaction cost of rate $\lambda$ (respectively, $\mu$). Thus in the
$n$--step binomial model a trading strategy with an
initial capital $x$
is a pair $\pi=(x,\{\gamma(k)\}_{k=1}^{n})$ where for any $k$,
$\gamma(k)$ is a random variable $\mathcal{F}^{\xi}_{k-1}$ measurable
which represents
the number of stocks that the investor has
at the moment $\frac{kT}{n}$, before the transfers
made in this moment of time.
The portfolio value (in cash) of a
trading strategy $\pi$ is given by
\begin{eqnarray}\label{2.12}
&V^{\pi}_{\lambda,\mu}(k)=x-(1+\lambda)
(\gamma(k)^{-}S^{\xi,n}(kT/n)+\\
&\sum_{i=1}^{k} (\gamma(i)-\gamma(i-1))^{+}
S^{\xi,n}((i-1)T/n))+(1-\mu)(\gamma(k)^{+}S^{\xi,n}(kT/n)+\nonumber\\
&\sum_{i=1}^{k}(\gamma(i)-
\gamma(i-1))^{-} S^{\xi,n}((i-1)T/n)),  \ \ k=0,1,...,n. \nonumber
\end{eqnarray}
Note that $V^\pi_{\lambda,\mu}(k)$ is the portfolio value
at the time $\frac{kT}{n}$ before the made transfers in this time.
A self financing
strategy $\pi$ is called \textit{admissible}
if the following no--bankruptcy condition holds
\begin{equation}\label{2.13}
{V}^{\pi}_{\lambda,\mu}(k)\geq{0} \ \ \ \ \ \forall k\leq n.
\end{equation}
The set of all \textit{admissible} self financing strategies with an initial
capital $x$ will denoted by $\mathcal{A}^{\xi,n}(x,\lambda,\mu)$.

Consider an American contingent claim
with the adapted payoff process
\begin{equation}\label{2.14}
Y^{\xi,n}(k)=F\big(\frac{kT}{n},S^{\xi,n}\big), \  \ 0 \leq k\leq n.
\end{equation}
For $\pi\in\mathcal{A}^{\xi,n}(x,\lambda,\mu)$
the shortfall risk defined by
\begin{equation}\label{2.15}
R_n(\pi,\lambda,u)=\max_{\tau\in\mathcal{T}^{\xi}_{0,n}}E^\xi_n[(Y^{\xi,n}({\tau})-
{V}^{\pi}_{\lambda,\mu}(\tau))^+]
\end{equation}
where $E^\xi_n$ is the expectation with respect to the probability measure
$P^\xi_n$. The shortfall risk for an initial capital $x$ is given by
\begin{equation}\label{2.16}
R_n(x,\lambda,\mu)=\inf_{\pi\in\mathcal{A}^{\xi,n}(x,\lambda,\mu)} R_n(\pi,\lambda,\mu).
\end{equation}

The following theorem is the main result of the paper and it says that the shortfall risk of an American option in the
BS market with proportional transaction costs $\lambda,\mu$ can be approximated by a sequence of shortfall risks
of an American options with same proportional costs in the binomial models defined above.
This result has a practical value since for any $n$ the shortfall risk $R_n(x,\lambda,\mu)$
can be calculated by dynamical programming algorithm which is given in Section 3.
\begin{thm}\label{thm2.1}
For any $\lambda>0$, $0<\mu<1$ and $x\in\mathbb{R}_{+}$,
\begin{equation}\label{2.17}
\lim_{n\rightarrow\infty}R_n(x,\lambda,\mu)=R(x,\lambda,\mu).
\end{equation}
\end{thm}
Next, we introduce a simple form of Skorohod embedding which allows
to consider the above binomial markets
and the BS model on the same probability space.
Set $W^{*}(t)=\frac{\ln S^W(t)}{\sigma}$, $t\geq 0$ \
and for any $n\in\mathbb{N}$ define
recursively
$\theta^{(n)}_0=0$, $\theta^{(n)}_{k+1}=\inf{\{t>\theta^{(n)}_k
:|W^{*}(t)-W^{*}(\theta^{(n)}_k)|=\sqrt\frac{T}{n}\}}$.
Observe (see Dolinsky and Kifer 2008) that for any $k$,
$W^{*}(\theta^{(n)}_{k+1})-W^{*}(\theta^{(n)}_k)$ is independent
of $\mathcal{F}^{W}_{\theta^{(n)}_k}$ and excepts the values $\sqrt\frac{T}{n}$ and
$-\sqrt\frac{T}{n}$, with probabilities $p^{(n)}$
and $1-p^{(n)}$, respectively. For any $n$ define the map
$\Pi_n:L^{\infty}(\mathcal{F}^{\xi}_n,P^{\xi}_n)\rightarrow{L^{\infty}
(\mathcal{F}^{W}_{\theta^{(n)}_n} ,P^W)}$
by $\Pi_n(U)=\tilde U$ so
that if $U=f\bigg(\sqrt{\frac{T}{n}}\xi_1,...,\sqrt{\frac{T}{n}}\xi_n\bigg)$
for a function $f$ on $\{\sqrt{\frac{T}{n}},-\sqrt{\frac{T}{n}}\}^n$
then $$\tilde
U=f(W^{*}({\theta^{(n)}_1}),W^{*}({\theta^{(n)}_2})-W^{*}({\theta^{(n)}_1}),...,
W^{*}({\theta^{(n)}_n})-W^{*}({\theta^{(n)}_{n-1}})).$$
Let $\mathcal{A}^{W,n}(x,\lambda,\mu)$ be set of
\textit{admissible} self financing strategies which
managed on the set $\{0,\theta^{(n)}_1,...,\theta^{(n)}_n\}$ such that after the moment $\theta^{(n)}_n$
the number of stocks in the portfolio is $0$.
Namely,
$\pi=(x,\{\gamma(t)\}_{t=0}^\infty)\in\mathcal{A}^{W,n}(x,\lambda,\mu)$
if there are random variables $u_1,...,u_n$ such that
\begin{equation}\label{2.17+}
\gamma(t)=\sum_{i=0}^{n-1} \mathbb{I}_{\theta^{(n)}_i< t\leq\theta^{(n)}_{i+1}}u_{i+1}
\end{equation}
where for any $i\geq 1$, $u_i$ is $\mathcal{F}^W_{\theta^{(n)}_{i-1}}$
measurable. We require that the corresponding wealth process which is given by
(\ref{2.6}) will satisfy the no--bankruptcy condition (\ref{2.7}).
The map $\Pi_n$ allows us to define a function
$\psi_n:\mathcal{A}^{\xi,n}(x,\lambda,\mu)\rightarrow\mathcal{A}^{W,n}(x,\lambda,\mu)$ which
maps \textit{admissible} self financing strategies
in the
$n$--step binomial model to the set of \textit{admissible}
self financing strategies in
the BS model. Let $\pi=(x,\{\gamma(k)\}_{k=1}^{n})
\in \mathcal{A}^{\xi,n}(x,\lambda,\mu)$. Define $\psi_n(\pi)=(x,\{\tilde\gamma(t)\}_{t=0}^\infty)$
where
\begin{equation}\label{2.18}
\tilde\gamma(t)=\sum_{i=0}^{n-1} \mathbb{I}_{\theta^{(n)}_i< t\leq\theta^{(n)}_{i+1}}\Pi_n(\gamma(i+1))
\end{equation}
where we set $\mathbb{I}_{A}=1$ if an event $A$ occurs and
$\mathbb{I}_{A}=0$ if not. Let us show that $\tilde\pi:=\psi_n(\pi)$ is an \textit{admissible} portfolio.
From (\ref{2.6}), (\ref{2.12}) and the equality $\Pi_n(S^{\xi,n}(kT/n))=S^W(\theta^{(n)}_k)$, $k\leq n$
it follows that
\begin{equation}\label{2.18+}
V^{\tilde\pi}_{\lambda,\mu}(\theta^{(n)}_k)=\Pi_n(V^{\pi}_{\lambda,\mu}(k))\geq 0,  \ \ k=0,1,...,n.
\end{equation}
The portfolio strategy $\tilde\pi$ is managed only on the
set $\{0,\theta^{(n)}_1,...,\theta^{(n)}_n\}$, and so
it is clear that the wealth process
${\{V^{\tilde\pi}_{\lambda,\mu}(t)\}}_{t=0}^\infty$
is a supermartingale with respect to the measure $\tilde{P}^W$
Furthermore for any $t$,
$V^{\tilde\pi}_{\lambda,\mu}(t)=V^{\tilde\pi}_{\lambda,\mu}(t\wedge\theta^{(n)}_n)$.
This together with (\ref{2.18+}) gives
\begin{equation}\label{2.18+++}
V^{\tilde\pi}_{\lambda,\mu}(t)\geq \tilde{E}^W(V^{\tilde\pi}_{\lambda,\mu}(\theta^{(n)}_n)|
\mathcal{F}^W_{\theta^{(n)}_n\wedge t})\geq 0.
\end{equation}
Thus $\psi_n(\pi)$
satisfies the no--bankruptcy condition, and $\psi_n(\pi)\in\mathcal{A}^{W,n}(x,\lambda,\mu)$.
If we restrict the portfolio $\psi_n(\pi)$ to the interval $[0,T]$ we obtain
an element which belongs to $\mathcal{A}^W(x,\lambda,\mu)$.

In Section 3 we prove that the optimal portfolios
for the shortfall risk measure
in the above binomial models can
be calculated by using a dynamical programming algorithm.
The following result shows how to use
these portfolios together with the maps
$\psi_n$, $n\in\mathbb{N}$ in order to construct "almost" optimal
portfolios in the BS model.
\begin{thm}\label{thm2.2}
Let $\lambda>0$, $0<\mu<1$ and $x>0$. For any $n\in\mathbb{N}$ let
$\pi_n=\pi_n(x,\lambda,\mu)\in\mathcal{A}^{\xi,n}(x,\lambda,\mu)$
be the optimal portfolio given by (\ref{2+.24}). Then
\begin{equation}\label{2.19}
\lim_{n\rightarrow\infty}R(\psi_n(\pi_n),\lambda,\mu)=R(x,\lambda,\mu).
\end{equation}
\end{thm}

\section{Analysis of the binomial models}\label{sec2+}\setcounter{equation}{0}
In this section we provide a dynamical programming algorithm for the shortfall risks and the corresponding optimal portfolios
in the binomial models. This dynamical programming algorithm will be essential
for comparing the shortfall risks in the binomial models
with the shortfall risk in the BS model. Through this section we will assume that the transaction costs $\lambda,\mu$
are fixed.

Let $\pi=(x,\{\gamma(k)\}_{k=1}^n)\in\mathcal{A}^{\xi,n}(x,\lambda,\mu)$ for some $x\geq 0$ and $n\in\mathbb{N}$.
From (\ref{2.12}) it follows that
\begin{eqnarray}\label{2+.1}
&V^\pi_{\lambda,\mu}(k+1)=G\big(V^\pi_{\lambda,\mu}(k),\gamma(k)S^{\xi,n}(kT/n),
(\gamma(k+1)-\gamma(k))\times\\
&S^{\xi,n}(kT/n),
\exp(\sigma\sqrt\frac{T}{n}\xi_{k+1})\big), \ \ k=0,1,...,n-1\nonumber
\end{eqnarray}
where
\begin{eqnarray}\label{2+.2}
&G(u,v,w,\rho)=u-(1-\mu)v^{+}+(1+\lambda)v^{-}+(1-\mu)w^{-}-\\
&(1+\lambda)w^{+}+\rho((1-\mu)(w+v)^{+}-(1+\lambda)(w+v)^{-}).\nonumber
\end{eqnarray}
For any $(u,v)\in\mathbb{R}_{+}\times\mathbb{R}$, $0<a<1$ and $b>0$
introduce the set $\mathcal{A}_{a,b}(u,v)=\{w|G(u,v,w,1+b), G(u,v,w,1-a)\geq 0\}$.
From simple calculations we obtain
\begin{eqnarray}\label{2+.3}
&\mathcal{A}_{a,b}(u,v)=\big[-v-\frac{u}{(1+\lambda)(1+b)-(1-\mu)},
\frac{(u-a v(1-\mu))^{+}}{1+\lambda-(1-\mu)(1-a)}-\frac{(u-a v(1-\mu))^{-}}{a(1-\mu)}\big]\\
&\mbox{if} \ v\geq 0 \ \mbox{and} \  \mathcal{A}_{a,b}(u,v)=\big[-\frac{(u+b(1+\lambda) v)^{+}}{(1+b)(1+\lambda)-(1-\mu)}+
\frac{(u+b(1+\lambda) v)^{-}}{b(1+\lambda)},\nonumber\\
&-v+\frac{u}{1+\lambda-(1-\mu)(1-a)}\big]
 \ \mbox{if} \ v<0.
\nonumber
\end{eqnarray}
Set $a_n=1-\exp(-\sigma\sqrt\frac{T}{n})$ and $b_n=\exp(\sigma\sqrt\frac{T}{n})-1$.
From (\ref{2+.1}) and the independency of $\xi_{k+1}$ and $\mathcal{F}^\xi_k$
it follows that $\pi=(x,\{\gamma(k)\}_{k=1}^n)\in\mathcal{A}^{\xi,n}(x,\lambda,\mu)$ iff
for any $k$, $\gamma(k)$ is $\mathcal{F}^\xi_{k-1}$ measurable and
\begin{equation}\label{2+.4}
(\gamma(k+1)-\gamma(k))S^{\xi,n}(kT/n)\in\mathcal{A}_{a_n,b_n}(V^{\pi}_{\lambda,\mu}(k),\gamma(k)S^{\xi,n}(kT/n)).
\end{equation}
Next, we prove a technical lemma.
\begin{lem}\label{lem2+.1}
Let $0<a,p<1$, $b>0$
and $H_1,H_2:\mathbb{R}_{+}\times\mathbb{R}\rightarrow\mathbb{R}_{+}$
be a functions which satisfy the following conditions. For i=1,2:\\
i. $H_i$ is a continuous function.\\
ii. For any $v\in\mathbb{R}$, $H_i(\cdot,v)$ is a non increasing function.\\
iii. $H_i$ is a piecewise linear function which vanishing at infinity
with respect to the first variable.
Namely, there exists a natural numbers $N^{(i)},M^{(i)}\in\mathbb{N}$ and a convex \textit{polyhedrals} $K^{(i)}_1,...,K^{(i)}_{N^{(i)}}\subset\mathbb{R}_{+}\times\mathbb{R}$ with pairwise disjoint interiors
and $\bigcup_{j=1}^{N^{(i)}}K^{(i)}_j=[0,M^{(i)}]\times\mathbb{R}$, such that for any $j\leq N^{(i)}$
\begin{equation}\label{2+.5}
H_i(u,v)=c^{(i)}_j u+d^{(i)}_j v+e^{(i)}_j \ \ \ \ \forall (u,v)\in K^{(i)}_j
\end{equation}
where $c^{(i)}_1,...,c^{(i)}_{N^{(i)}},d^{(i)}_1,...,d^{(i)}_{N^{(i)}},e^{(i)}_1,...,e^{(i)}_{N^{(i)}}\in\mathbb{R}$
are constants.\\

Define the function $H:\mathbb{R}_{+}\times\mathbb{R}\rightarrow\mathbb{R}_{+}$ by
\begin{eqnarray}\label{2+.10}
&H(u,v)=\inf_{w\in\mathcal{A}_{a,b}{(u,v)}}pH_1\big(G(u,v,w,1+b),(v+w)(1+b)\big)+\\
&(1-p)H_1\big(G(u,v,w,1-a),(v+w)(1-a)\big).\nonumber
\end{eqnarray}
Then $H$ is satisfying the conditions i.--iii. above.
\end{lem}
\begin{proof}
Set, $I(u,v,w)=pH_1\big(G(u,v,w,1+b),(v+w)(1+b)\big)+(1-p)H_1\big(G(u,v,w,1-a),(v+w)(1-a)\big).$
Observe that $I(\cdot,u,v)$ is a non increasing function for any $v,w$. Clearly, for any $0 \leq u_1<u_2$ and $v\in\mathbb{R}$,
$\mathcal{A}_{a,b}(u_1,v)\subseteq\mathcal{A}_{a,b}(u_2,v)$.
Thus,
\begin{eqnarray}\label{2+.12}
&H(u_1,v)=\inf_{w\in\mathcal{A}_{a,b}{(u_1,v)}}I(u_1,v,w)
\geq
\inf_{w\in\mathcal{A}_{a,b}{(u_2,v)}}I(u_1,v,w)\geq \\
&\inf_{w\in\mathcal{A}_{a,b}{(u_2,v)}}I(u_2,v,w)=H(u_2,v)\nonumber
\end{eqnarray}
and so, $H$ satisfies condition ii. Next, we prove continuity.
Let $(u,v)\in\mathbb{R}_{+}\times\mathbb{R}$ and
$\{(u_n,v_n)\}_{n=1}^\infty\subset\mathbb{R}_{+}\times\mathbb{R}$
such that $(u_n,v_n)\rightarrow (u,v)$ and $\lim_{n\rightarrow\infty} H(u_n,v_n)$ exists (may be $\pm\infty$).
For any $n$ there exists ($I$ is a continuous function)
$w_n\in\mathcal{A}_{a,b}(u_n,v_n)$ which satisfies
$I(u_n,v_n,w_n)=H(u_n,v_n)$. The sequence $\{w_n\}_{n=1}^\infty$ is bounded and so its has a subsequence
$\{w_{n_k}\}_{k=1}^\infty$ which converge to $w$.
From (\ref{2+.3}) it follows that $w\in\mathcal{A}_{a,b}(u,v)$ and so
\begin{equation}\label{2+.13}
H(u,v)\leq I(u,v,w)=\lim_{n\rightarrow\infty} I(u_n,v_n,w_n)=\lim_{n\rightarrow\infty} H(u_n,v_n).
\end{equation}
Choose $\tilde w \in \mathcal{A}_{a,b}(u,v)$ for which $I(u,v,\tilde w)=H(u,v)$. From (\ref{2+.3}) it follows that
there exists a sequence $\tilde w_n \in\mathcal{A}_{a,b}(u_n,v_n)$, $n\in\mathbb{N}$
such that $\lim_{n\rightarrow\infty}\tilde w_n =\tilde w.$ Thus,
\begin{equation}\label{2+.14}
H(u,v)=I(u,v,\tilde w)=\lim_{n\rightarrow\infty} I(u_n,v_n,\tilde w_n)\geq\lim_{n\rightarrow\infty} H(u_n,v_n).
\end{equation}
From (\ref{2+.13})--(\ref{2+.14}) we obtain that $H$ is continuous. Finally,
we prove that $H$ satisfies condition iii.
For any $(u,v)\in\mathbb{R}_{+}\times\mathbb{R}$ introduce the set
\begin{eqnarray*}
&B(u,v)=\big\{w\in \mathcal{A}_{a,b}(u,v)|\big(G(u,v,w,1+b),(v+w)(1+b)\big)\in \bigcup_{j=1}^{N^{(1)}}\partial K^{(1)}_j\big\}\bigcup\\
&\big\{w\in \mathcal{A}_{a,b}(u,v)|\big(G(u,v,w,1-b),(v+w)(1-b)\big)\in \bigcup_{j=1}^{N^{(2)}}\partial K^{(2)}_j\big\}\bigcup \partial \mathcal{A}_{a,b}(u,v)
.\nonumber
\end{eqnarray*}
Fix $u,v$ and let $B(u,v)=\{w_1<w_2<...<w_k\}$. From (\ref{2+.5}) it follows that for any $i<k$, the function
$I(u,v,\cdot)$ is linear on the interval $[w_i,w_{i+1}]$ and so
\begin{equation}\label{2+.16}
H(u,v)=\min_{w\in B(u,v)}I(u,v,w).
\end{equation}
Note that there exists a finite sequence of real numbers $\alpha_1,...,\alpha_N,\beta_1,...,\beta_N,\delta_1,...,\delta_N$
such that for any $(u,v)$,
$B(u,v)\subseteq\{\alpha_j u+\beta_j v+\delta_j|j\leq N\}.$
This together with (\ref{2+.16}) gives that there
there exists a finite sequence of real numbers $\Phi_1,...,\Phi_m,\Delta_1,...,\Delta_m,$
$\Theta_1,...,\Theta_m$
such that for any $(u,v)\in\mathbb{R}_{+}\times\mathbb{R}$
\begin{equation}\label{2+.17}
H(u,v)=\Phi_j u+\Delta_j v+\Theta_j
\end{equation}
for some $j$ (which depends on $(u,v))$. From (\ref{2+.3}),
$-v\in\mathcal{A}_{a,b}(u,v)$ and so
\begin{equation}\label{2+.18}
H(u,v)\leq I(u,v,-v)=p H_1(u,0)+(1-p)H_2(u,0)\leq\max(H_1(u,0),H_2(u,0)).
\end{equation}
From (\ref{2+.17})--(\ref{2+.18}) and the fact that $H$ is continuous we conclude that
$H$ satisfies condition iii. and the proof is completed.
\end{proof}

Next, fix $n$ and consider the $n$--step binomial model.
For any $\pi\in\mathcal{A}^{\xi,n}(x,\lambda,\mu)$ define a sequence of
random variables ${\{U^\pi(k)\}}_{k=0}^n$ by
\begin{eqnarray}\label{2+.19}
&U^\pi(n)=({Y}^{\xi,n}(n)-{V}^\pi_{\lambda,\mu}(n))^+, \ \mbox{and} \ \mbox{for} \ k<n\\
&U^\pi(k)=\max\big(E^{\xi}_n(U^{\pi}({k+1})
|\mathcal{F}^{\xi}_k),({Y}^{\xi,n}(k)-V^\pi_{\lambda,\mu}(k))^{+}\big).\nonumber
\end{eqnarray}
Applying standard results for optimal stopping (see Peskir and Shiryaev 2006)
for the process $({Y}^{\xi,n}(k)-{V}^{\pi}_{\lambda,\mu}(k))^{+}$, $k=0,1,...,n$
we obtain
\begin{equation}\label{2+.20}
U^\pi(0)=
\max_{\tau\in{\mathcal{T}^\xi_{0,n}}}E^{\xi}[(Y^{\xi,n}({\tau})-{V}^{\pi}_{\lambda,\mu}({\tau}))^+]=R_n(\pi,\lambda,\mu).
\end{equation}
For any $0\leq k\leq n$ let $\phi^{(n)}_k:\{-1,1\}^k\rightarrow\mathbb{R}_{+}$ such that
\begin{equation}\label{2+.21}
\phi^{(n)}_k(\xi_1,...,\xi_k)=Y^{\xi,n}(k).
\end{equation}
Define a sequence of functions
$J^{(n)}_k:\mathbb{R}_{+}\times\mathbb{R}\times \{-1,1\}^k\rightarrow \mathbb{R}_{+},\, k=0,1,...,
n$ by the following backward relations. For any $z_1,...,z_n\in\{-1,1\}$ and $(u,v)\in\mathbb{R}_{+}\times\mathbb{R}$
\begin{eqnarray}\label{2+.22}
&J^{(n)}_n(u,v,z_1,...,z_n)=(\phi^{(n)}_n(z_1,...,z_n)-u)^+ \  \ \mbox{and}\\
&J^{(n)}_k(u,v,z_1,...,z_k)=\max\bigg((\phi^{(n)}_n(z_1,...,z_k)-u)^{+}, \ \
\inf_{w\in \mathcal{A}_{a_n,b_n}(u,v)}\nonumber\\
&p^{(n)} J^{(n)}_{k+1}\big(G(u,v,w,1+b_n),(1+b_n)(u+w),z_1,...,z_k,1\big)+(1-p^{(n)})\times\nonumber\\
&J^{(n)}_{k+1}\big(G(u,v,w,1-a_n),(1-a_n)(u+w),z_1,...,z_k,-1\big)\bigg)
\ \mbox{for}  \ k<n\nonumber
\end{eqnarray}
where recall, $p^{(n)}$ was defined after (\ref{2.11}).
From Lemma \ref{lem2+.1} it follows (by backward induction) that for any $k\leq n$ and
$z_1,...,z_k\in\{-1,1\}$ the function
$H(\cdot,\cdot):=
J^{(n)}_k(\cdot,\cdot,z_1,...,z_k)$ is satisfying conditions i.--iii
which were introduced in Lemma \ref{lem2+.1}. In particular it is continuous. This fact allows us to define the functions
$h^{(n)}_k:\mathbb{R}_{+}\times\mathbb{R}\times\{-1,1\}^k\rightarrow \mathbb{R}$, $k<n$
by
\begin{eqnarray}\label{2+.23}
&h^{(n)}_k(u,v,z_1,...,z_k)=argmin_{w\in\mathcal{A}_{a_n,b_n}(u,v)}\\
&p^{(n)} J^{(n)}_{k+1}\big(G(u,v,w,1+b_n),(1+b_n)(u+w),z_1,...,z_k,1\big)+(1-p^{(n)})\times\nonumber\\
&J^{(n)}_{k+1}\big(G(u,v,w,1-a_n),(1-a_n)(u+w),z_1,...,z_k,-1\big)\bigg).\nonumber
\end{eqnarray}
Let $x>0$ be an initial capital. Define $\pi=\pi_n(x,\lambda,\mu)=(x,\{\gamma(k)\}_{k=1}^n)$
by
\begin{eqnarray}\label{2+.24}
&V^{\pi}_{\lambda,\mu}(0)=x, \ \mbox{and} \ \mbox{for} \ 0\leq k<n, \ \gamma(k+1)=\gamma(k)+\\
&\frac{1}{S^{\xi,n}(kT/n)}h^{(n)}_k\big(V^{\pi}_{\lambda,\mu}(k),\gamma(k)S^{\xi,n}(kT/n),\xi_1,...,\xi_k \big) \  \mbox{and} \
V^{\pi}_{\lambda,\mu}(k+1)=\nonumber\\
&G\big(V^{\pi}_{\lambda,\mu}(k),\gamma(k)S^{\xi,n}(kT/n), (\gamma(k+1)-\gamma(k))S^{\xi,n}(kT/n),\exp(\sigma\sqrt\frac{T}{n}\xi_{k+1})\big).\nonumber
\end{eqnarray}
\begin{prop}\label{prop2+.2}
For any $n\in\mathbb{N}$ and $x\geq 0$
\begin{equation}\label{2+.25}
R_n(\pi_n(x,\lambda,\mu),\lambda,\mu)=R_n(x,\lambda,\mu)=J^{(n)}_0(x,0).
\end{equation}
\end{prop}
\begin{proof}
Fix $n\in\mathbb{N}$ and $x\geq0$. Set $\pi=\pi_n(x,\lambda,\mu)=(x,{\gamma})$ and let
$\tilde\pi=(x,\tilde\gamma)\in\mathcal{A}^{\xi,n}(x,\lambda,\mu)$ an arbitrary
portfolio.
First we prove by backward induction that
for any $k\leq n$,
\begin{eqnarray}\label{2+.26}
&J^{(n)}_k(V^{\pi}_{\lambda,\mu}(k),\gamma(k)S^{\xi,n}(kT/n),\xi_1,...,\xi_k)
=U^{\pi}(k)
\ \mbox{and} \\
&J^{(n)}_k(V^{\tilde\pi}_{\lambda,\mu}(k),\tilde\gamma(k)S^{\xi,n}(kT/n),\xi_1,...,\xi_k)
\leq U^{\tilde\pi}(k).\nonumber
\end{eqnarray}
For $k=n$,
we obtain from (\ref{2+.19}) and (\ref{2+.21})--(\ref{2+.22})
that the relations (\ref{2+.26}) hold with equality. Suppose that
(\ref{2+.26}) holds true for $k+1$ and prove them for $k$.
Set,
\begin{eqnarray*}
&\Upsilon=\gamma(k)S^{\xi,n}(kT/n), \  \tilde\Upsilon=\tilde\gamma(k)S^{\xi,n}(kT/n),\\
&\Gamma=h^{(n)}_k(V^{\pi}_{\lambda,\mu}(k),\Upsilon,\xi_1,...,\xi_k)
\ \mbox{and} \
\tilde\Gamma=(\tilde\gamma(k+1)-\tilde\gamma(k))S^{\xi,n}(kT/n).
\end{eqnarray*}
From (\ref{2+.23})--(\ref{2+.24}) and the induction assumption it follows
\begin{eqnarray}\label{2+.27}
&E^{\xi}_n(U^\pi(k+1)|\mathcal{F}^\xi_k)=E^{\xi}_n\bigg(J^{(n)}_{k+1}\bigg(G\big(V^{\pi}_{\lambda,\mu}(k),\Upsilon,
\Gamma,\exp(\sigma\sqrt{T/n}\xi_{k+1})\big),\\
&(\Gamma+\Upsilon)\exp(\sigma\sqrt{T/n}\xi_{k+1}),\xi_1,...,\xi_{k+1}\bigg)
\bigg|\mathcal{F}^\xi_k\bigg)=p^{(n)} J^{(n)}_{k+1}\bigg(G\big(V^{\pi}_{\lambda,\mu}(k),\Upsilon,\nonumber\\
&\Gamma,1+b_n\big),(\Gamma+\Upsilon)(1+b_n),\xi_1,...,\xi_k,1\bigg)
+(1-p^{(n)})J^{(n)}_{k+1}\bigg(G\big(V^{\pi}_{\lambda,\mu}(k),\Upsilon,\nonumber\\
&\Gamma,1-a_n\big),(\Gamma+\Upsilon)(1-a_n),\xi_1,...,\xi_k,-1\bigg)=
\min_{w\in\mathcal{A}_{a_n,b_n}(V^{\pi}_{\lambda,\mu}(k),\Upsilon)}\nonumber\\
&=p^{(n)} J^{(n)}_{k+1}\bigg(G\big(V^{\pi}_{\lambda,\mu}(k),\Upsilon,w,1+b_n\big),(w+\Upsilon)(1+b_n),\xi_1,...,\xi_k,1\bigg)+\nonumber\\
&(1-p^{(n)}) J^{(n)}_{k+1}\bigg(G\big(V^{\pi}_{\lambda,\mu}(k),\Upsilon,w,1-a_n\big),(w+\Upsilon)(1-a_n),\xi_1,...,\xi_k,-1\bigg).\nonumber
\end{eqnarray}
From (\ref{2+.4}) it follows that $\tilde\Gamma\in\mathcal{A}_{a_n,b_n}(V^{\tilde\pi}_{\lambda,\mu}(k),\tilde\Upsilon)$,
and so from the induction assumption
\begin{eqnarray}\label{2+.28}
&E^{\xi}_n(U^{\tilde\pi}(k+1)|\mathcal{F}^\xi_k)\geq E^{\xi}_n\bigg(J^{(n)}_{k+1}\bigg(G\big(V^{\tilde\pi}_{\lambda,\mu}(k),\tilde\Upsilon,
\tilde\Gamma,\exp(\sigma\sqrt{T/n}\xi_{k+1})\big),\\
&(\tilde\Gamma+\tilde\Upsilon)\exp(\sigma\sqrt{T/n}\xi_{k+1}),\xi_1,...,\xi_{k+1}\bigg)
\bigg|\mathcal{F}^\xi_k\bigg)=p^{(n)} J^{(n)}_{k+1}\bigg(G\big(V^{\tilde\pi}_{\lambda,\mu}(k),\tilde\Upsilon,\nonumber\\
&\tilde\Gamma,1+b_n\big),(\tilde\Gamma+\tilde\Upsilon)(1+b_n),\xi_1,...,\xi_k,1\bigg)
+(1-p^{(n)})J^{(n)}_{k+1}\bigg(G\big(V^{\tilde\pi}_{\lambda,\mu}(k),\tilde\Upsilon,\nonumber\\
&\tilde\Gamma,1-a_n\big),(\tilde\Gamma+\tilde\Upsilon)(1-a_n),\xi_1,...,\xi_k,-1\bigg)\geq
\min_{w\in\mathcal{A}_{a_n,b_n}(V^{\tilde\pi}_{\lambda,\mu}(k),\tilde\Upsilon)}\nonumber\\
&=p^{(n)} J^{(n)}_{k+1}\bigg(G\big(V^{\tilde\pi}_{\lambda,\mu}(k),\tilde\Upsilon,w,1+b_n\big),(w+\tilde\Upsilon)(1+b_n),\xi_1,...,\xi_k,1\bigg)+\nonumber\\
&(1-p^{(n)}) J^{(n)}_{k+1}\bigg(G\big(V^{\tilde\pi}_{\lambda,\mu}(k),\tilde\Upsilon,w,1-a_n\big),(w+\tilde\Upsilon)(1-a_n),\xi_1,...,\xi_k,-1\bigg).\nonumber
\end{eqnarray}
Combining (\ref{2+.19}), (\ref{2+.21})--(\ref{2+.22}) and (\ref{2+.27})-(\ref{2+.28}) we obtain that
(\ref{2+.26}) holds true. Next, by using
(\ref{2+.26}) for $k=0$ and (\ref{2+.20}) it follows that for any $\tilde\pi\in\mathcal{A}^{\xi,n}(x,\lambda,\mu)$
\begin{equation*}
R_n(\pi,\lambda,\mu)=U^\pi(0)=J^{(n)}_0(x,0)\leq U^{\tilde\pi}(0)=R_n(\tilde\pi,\lambda,\mu).
\end{equation*}
Thus $R_n(x,\lambda,\mu)=R_n(\pi,\lambda,\mu)=J^{(n)}_0(x,0)$, as required.
\end{proof}
\begin{cor}
From Lemma \ref{lem2+.1} and Proposition \ref{prop2+.2} we obtain that the function $R_n(x,\lambda,\mu)=J^{(n)}_0(x,0)$
is a continuous non increasing piecewise linear function vanishing at $\infty$. Namely, there exists a natural
 number $N$, $c_1,...,c_N\leq 0$, $d_1,...,d_N\in\mathbb{R}$ and
$0=\alpha_1<\alpha_2<...<\alpha_{N+1}<\infty$
 such that
$R_n(x,\lambda,\mu)=\sum_{i=1}^N \mathbb{I}_{[a_i,a_{i+1})}(c_ix+d_i)$.
\end{cor}
\section{Proof of the limit theorems}\label{sec3}
In this section we complete the proof of Theorems \ref{thm2.1}--\ref{thm2.2}. We start with a technical preparations.

For any $n\in\mathbb{N}$ set
\begin{equation}\label{3.1}
S^{W,n}(t)=S^W(\theta^{(n)}_k), \  \ kT/n\leq t<(k+1)T/n, \ \ k=0,1,...,n.
\end{equation}
Define
\begin{equation}\label{3.1+}
Y^{W,n}(t)=F(t,S^{W,n}), \ \ t\in [0,T].
\end{equation}
Note that for any $0\leq k\leq n$
\begin{eqnarray}\label{3.1++}
&Y^{W,n}(kT/n)=\phi^{(n)}_k\bigg(\sqrt\frac{n}{T}W^{*}(\theta^{(n)}_1),\sqrt\frac{n}{T}(W^{*}(\theta^{(n)}_2)-W^{*}({\theta^{(n)}_1})),
...,\\
&\sqrt\frac{n}{T}(W^{*}(\theta^{(n)}_k)-W^{*}({\theta^{(n)}_{k-1}}))\bigg).\nonumber
\end{eqnarray}
From Kifer (2006)
\begin{eqnarray}\label{3.1+++}
&\lim_{n\rightarrow\infty}E^W\sup_{0\leq t\leq T}|S^{W,n}(t)-S^W(t)|=0
\ \ \mbox{and}\\
&\lim_{n\rightarrow\infty}E^W\max_{1\leq{k}\leq{n}}|\theta^{(n)}_k-\frac{kT}{n}|
=0.\nonumber
\end{eqnarray}
Fix $n$. Following Kifer (2006) we
introduce for each $k=1,2,...$ the finite $\sigma$-algebra
$\mathcal{G}^{W,n}_k=\sigma\{W^{*}(\theta^{(n)}_1),...,W^{*}(\theta^{(n)}_k)\}$ with
$\mathcal{G}^{W,n}_0={\{\emptyset,\Omega_W\}}$ being the trivial
$\sigma$-algebra. Let $\mathcal{S}^{W,n}_{0,n}$ and $\mathcal{T}^{W,n}_{0,n}$, be the sets of all
stopping times with values in the set $\{0,1,...,n\}$ with
respect to the filtrations $\{\mathcal{G}^{W,n}_k\}_{k=0}^n$ and
$\{\mathcal{F}^W_{\theta^{(n)}_k}\}_{k=0}^n$, respectively.
Recall the set $\mathcal{A}^{W,n}(x,\lambda,\mu)$ which was introduced before equation
(\ref{2.17+}). Define
\begin{equation}\label{3.2}
R^{W,n}(x,\lambda,\mu)=\inf_{\pi\in\mathcal{A}^W(x,\lambda,\mu)}\sup_{\tau\in\mathcal{T}^{W,n}_{0,n}}
E^W[(Y^{W,n}(\tau T/n)-V^\pi_{\lambda,\mu}(\theta^{(n)}_{\tau}))^{+}].
\end{equation}
From (\ref{2.6}) it follows that for any $\pi=(x,\{\gamma(t)\}_{t=0}^\infty)\in\mathcal{A}^{W,n}(x,\lambda,\mu)$,
\begin{equation}\label{3.2+}
V^\pi_{\lambda,\mu}(\theta^{(n)}_{k+1})=G\big(V^\pi_{\lambda,\mu}(k),\gamma(\theta^{(n)}_k) S^W(\theta^{(n)}_k),\Upsilon,
\exp(\sigma(W^{*}(\theta^{(n)}_{k+1})-W^{*}(\theta^{(n)}_k)))\big)
\end{equation}
where $\Upsilon=(\gamma(\theta^{(n)}_{k+1})-\gamma(\theta^{(n)}_k))S^W(\theta^{(n)}_k)$ and $G$ was introduced
in (\ref{2+.2}).

Combining similar arguments to those of Section 3 (replace
$\{\xi_i\}_{i=1}^n$, $\{S^{\xi,n}(\frac{iT}{n})\}_{i=0}^n$, and
$\{\mathcal{F}^\xi_i\}_{i=0}^n$
by
$\{\sqrt\frac{n}{T}(W^{*}(\theta^{(n)}_i)-W^{*}({\theta^{(n)}_{i-1}}))\}_{i=1}^n$,
$\{S^W(\theta^{(n)}_i)\}_{i=0}^n$ and $\{\mathcal{F}^W_{\theta^{(n)}_i}\}_{i=0}^n$, respectively) with
(\ref{3.1++}), (\ref{3.2+}) and the independency of $W^{*}(\theta^{(n)}_{k+1})-W^{*}(\theta^{(n)}_k)$ and
$\mathcal{F}^{W}_{\theta^{(n)}_k}$, we obtain
\begin{equation}\label{3.2++}
R^{W,n}(x,\lambda,\mu)=J^{(n)}_0(x)=R_n(x,\lambda,\mu) \ \ \forall{x,\lambda,\mu}.
\end{equation}

Next, fix an initial capital $x$ and a proportional transaction costs
$\lambda,\mu$. For any $n$ let $\pi_n=\pi_n(x,\lambda,\mu)$ be the optimal portfolio which is given by
(\ref{2+.24}). Consider the portfolio $\tilde\pi_n:=\psi_n(\pi_n)\in\mathcal{A}^W(x,\lambda,\mu)$.
For these portfolios we have the following lemma.
\begin{lem}\label{lem3.1}
\begin{equation}\label{3.3}
\lim \sup_ {n\rightarrow\infty}R(\tilde\pi_n,\lambda,\mu)-R_n(x,\lambda,\mu)\leq 0.
\end{equation}
\end{lem}
\begin{proof}
For any $n$ let $\tau_n\in\mathcal{T}^W_{[0,T]}$ such that
\begin{equation}\label{3.3+}
R(\tilde\pi_n,\lambda,\mu)<\frac{1}{n}+E^W[(Y^W(\tau_n)-{V}^{\tilde\pi_n}_{\lambda,\mu}(\tau_n))^{+}].
\end{equation}
Define $\nu_n=n\wedge\min\{k|\theta^{(n)}_k\geq \tau_n\}$, $n\in\mathbb{N}$.
Observe that $\nu_n\in\mathcal{T}^{W,n}_{0,n}$ and $\theta^{(n)}_{\nu_n}\geq\tau_n\wedge\theta^{(n)}_n$.
The portfolio value process $\{V^{\tilde\pi_n}_{\lambda,\mu}(t)\}_{t=0}^\infty$ is a
supermartingale with respect to the measure $\tilde{P}^W$.
Thus for any $n\in\mathbb{N}$
\begin{equation}\label{3.4}
V^{\tilde\pi_n}_{\lambda,\mu}(\tau_n)=V^{\tilde\pi_n}_{\lambda,\mu}(\tau_n\wedge\theta^{(n)}_n)\geq
\tilde{E}^W(V^{\tilde\pi_n}_{\lambda,\mu}(\theta^{(n)}_{\nu_n})|\mathcal{F}^W_{\tau_n\wedge\theta^{(n)}_n}).
\end{equation}
From (\ref{2.2+}), (\ref{3.4}) and the Jensen inequlity it follows
\begin{eqnarray}\label{3.5}
&E^W[(Y^W(\tau_n\wedge\theta^{(n)}_n)-{V}^{\tilde\pi_n}_{\lambda,\mu}(\tau_n))^{+}]=
\tilde E^W \bigg(\frac{1}{Z(\tau_n\wedge\theta^{(n)}_n)}\big(Y^W(\tau_n\wedge\theta^{(n)}_n)\\
&-{V}^{\tilde\pi_n}_{\lambda,\mu}(\tau_n)\big)^{+}\bigg)\leq \tilde E^W \bigg(\frac{1}{Z(\tau_n\wedge\theta^{(n)}_n)}\big(Y^W(\tau_n\wedge\theta^{(n)}_n)
-{V}^{\tilde\pi_n}_{\lambda,\mu}(\theta^{(n)}_{\nu_n})\big)^{+}\bigg)\nonumber\\
&= E^W \bigg(\frac{Z(\theta^{(n)}_{\nu_n})}{Z(\tau_n\wedge\theta^{(n)}_n)}\big(Y^W(\tau_n\wedge\theta^{(n)}_n)
-{V}^{\tilde\pi_n}_{\lambda,\mu}(\theta^{(n)}_{\nu_n})\big)^{+}\bigg).\nonumber
\end{eqnarray}
From (\ref{3.3+}) and (\ref{3.5}),
\begin{eqnarray}\label{3.6}
&R(\tilde\pi_n,\lambda,\mu)<\frac{1}{n}+E^W|Y^W(\tau_n)-Y^W(\tau_n\wedge\theta^{(n)}_n)|+   \\
&E^W\big(\big|\frac{Z(\theta^{(n)}_{\nu_n})}{Z(\tau_n\wedge\theta^{(n)}_n)}-1\big|\sup_{0\leq t\leq T}Y^W(t)\big)+
E^W|Y^W(\tau_n\wedge\theta^{(n)}_n)-Y^{W,n}(\nu_n T/n)|\nonumber \\
&+E^W[(Y^{W,n}(\nu_n T/n)-{V}^{\tilde\pi_n}_{\lambda,\mu}(\theta^{(n)}_{\nu_n}))^{+}]. \nonumber
\end{eqnarray}
From the definition it follows
\begin{eqnarray*}
&\tau_n-\tau_n\wedge\theta^{(n)}_n\leq |T-\theta^{(n)}_n|, \
\theta^{(n)}_{\nu_n}-\tau_n\wedge\theta^{(n)}_n\leq \max_{0\leq k<n}\theta^{(n)}_{k+1}-\theta^{(n)}_k\leq \frac{T}{n}+\\
&2\max_{1\leq{k}\leq{n}}|\theta^{(n)}_k-\frac{kT}{n}| \  \mbox{and} \
|\tau_n\wedge\theta^{(n)}_n-\nu_n T/n|\leq \frac{T}{n}+\max_{1\leq{k}\leq{n}}|\theta^{(n)}_k-\frac{kT}{n}|.
\end{eqnarray*}
From (\ref{3.1+++}) we get that the sequences
$\{\tau_n-\tau_n\wedge\theta^{(n)}_n\}_{n=1}^\infty$,
$\{\theta^{(n)}_{\nu_n}-\tau_n\wedge\theta^{(n)}_n\}_{n=1}^\infty$
and $\{\tau_n\wedge\theta^{(n)}_n-\nu_n T/n\}_{n=1}^\infty$ converge to $0$ in probability. From (\ref{3.1+++})
$S^{W,n}\rightarrow S^W$ (on the space $M[0,T]$) in probability. Since $F$ is continuous and the process $Z$ is
continuous we obtain that the sequences
$\{Y^W(\tau_n)-Y^W(\tau_n\wedge\theta^{(n)}_n)\}_{n=1}^\infty$,
$\{|\frac{Z(\theta^{(n)}_{\nu_n})}{Z(\tau_n\wedge\theta^{(n)}_n)}-1\big|\sup_{0\leq t\leq T}Y^W(t)\}_{n=1}^\infty$
and $\{Y^W(\tau_n\wedge\theta^{(n)}_n)-Y^{W,n}(\nu_n T/n)\}_{n=1}^\infty$ converge to $0$ in probability.
From (\ref{2.3}) it follows that the above sequences are uniformly integrable,
and so they converge to $0$ in $L^1(\Omega_W,P^W)$. Thus from (\ref{3.6})
\begin{equation}\label{3.6+}
\lim \sup_{n\rightarrow\infty} R(\tilde\pi_n,\lambda,\mu)-A_n\leq 0
\end{equation}
where $A_n=E^W[(Y^{W,n}(\nu_n T/n)-{V}^{\tilde\pi_n}_{\lambda,\mu}(\theta^{(n)}_{\nu_n}))^{+}]$.
Note that the process
$\{(Y^{W,n}(kT/n)$
$-{V}^{\tilde\pi_n}_{\lambda,\mu}(\theta^{(n)}_k))^{+}\}_{k=0}^n$
is adapted to the filtration $\{\mathcal{G}^{W,n}_k\}_{k=0}^n$, thus from standard dynamical programming
(see Peskir and Shiryaev 2006)
it follows
\begin{eqnarray}\label{3.7}
&A_n\leq\sup_{\zeta\in\mathcal{T}^{W,n}_{0,n}}
E^W[(Y^{W,n}(\zeta T/n)-{V}^{\tilde\pi_n}_{\lambda,\mu}(\theta^{(n)}_{\zeta}))^{+}]=\\
&\sup_{\zeta\in\mathcal{S}^{W,n}_{0,n}}
E^W[(Y^{W,n}(\zeta T/n)-{V}^{\tilde\pi_n}_{\lambda,\mu}(\theta^{(n)}_{\zeta}))^{+}].\nonumber
\end{eqnarray}
Recall the map $\Pi_n$ which was introduced after Theorem \ref{thm2.1}.
Notice that $\Pi_n:\mathcal{T}^\xi_{0,n}\rightarrow\mathcal{T}^{S,n}_{0,n}$ is a bijection
and for any random variable $U\in L^{\infty}(\mathcal{F}^{\xi}_n,P^{\xi}_n)$,
$E^W\Pi_n(U)=E^\xi_n U$. From (\ref{2.18+}), (\ref{2+.21}) and (\ref{3.1++})
we obtain
\begin{eqnarray}\label{3.9}
&\sup_{\zeta\in\mathcal{T}^{S,n}_{0,n}}
E^W[(Y^{W,n}(\zeta T/n)-{V}^{\tilde\pi_n}_{\lambda,\mu}(\theta^{(n)}_{\zeta}))^{+}]=\\
&\sup_{\sigma\in\mathcal{T}^{\xi}_{0,n}}E^W \big(\Pi_n[(Y^{\xi,n}(\sigma )
-{V}^{\pi_n}_{\lambda,\mu}(\sigma))^{+}]\big)=\nonumber\\
&\sup_{\sigma\in\mathcal{T}^{\xi}_{0,n}}E^\xi_n[(Y^{\xi,n}(\sigma )-{V}^{\pi_n}_{\lambda,\mu}(\sigma))^{+}]=R_n(\pi_n,\lambda,\mu)=
R_n(x,\lambda,\mu).
\nonumber
\end{eqnarray}
By combining (\ref{3.6+})--(\ref{3.9}) we complete the proof.
\end{proof}
Let $\lambda>0$ and $0<\mu<1$. Set $\lambda_n=(1+\lambda)\exp(-2\sigma\sqrt\frac{T}{n})-1$ and $\mu_n=1-(1-\mu)\exp(2\sigma\sqrt\frac{T}{n})$
(we assume that $n$ is sufficiently large such that $\lambda_n>0$ and $0<\mu_n<1$).
\begin{lem}\label{lem3.2}
For any initial capital $x$,
\begin{equation}\label{3.10}
\lim \sup_ {n\rightarrow\infty} R_n(x,\lambda_n,\mu_n)\leq R(x,\lambda,\mu).
\end{equation}
\end{lem}
\begin{proof}
Choose $\epsilon>0$. There exists $\pi=(x,\{\gamma(t)\}_{t=0}^T)\in\mathcal{A}^W(x,\lambda,\mu)$
such that
\begin{equation}\label{3.10+}
R(\pi,\lambda,\mu)<\epsilon+R(x,\lambda,\mu).
\end{equation}
For simplicity we extend the portfolio $\pi$ to $\mathbb{R}_{+}$, by setting
$\gamma(t)=0$ for $t>T$, i.e. the portfolio value remains constant after the maturity date $T$.
Set $u_n(k)=\gamma(\theta^{(n)}_k)$, $n\in\mathbb{N}$, $0\leq k\leq n$.
For any $n$ define the adapted (to the filtration $\{\mathcal{F}^W_t\}_{t=0}^\infty$) process
$\{\gamma_n(t)\}_{t=0}^\infty$ by
\begin{equation}\label{3.10+-}
\gamma_n(t)=\sum_{k=0}^{n-1} \mathbb{I}_{\theta^{(n)}_k<t\leq\theta^{(n)}_{k+1}}u_n(k).
\end{equation}
Consider the portfolio $\pi_n=(x,\{\gamma_n(t)\}_{t=0}^\infty)$
in a BS model for which purchase and sale, of
the risky asset are subject to a proportional
transaction costs of rate $\lambda_n$ and
$\mu_n$, respectively.
Observe that for any $i<n$ we have the inequalities
$\exp(2\sigma\sqrt\frac{T}{n})\inf_{\theta^{(n)}_i\leq t\leq \theta^{(n)}_{i+1}} S^W(t)\geq
S^W(\theta^{(n)}_{i+1})$ and $\exp(-2\sigma\sqrt\frac{T}{n})\sup_{\theta^{(n)}_i\leq t\leq \theta^{(n)}_{i+1}} S^W(t)$\\$\leq
S^W(\theta^{(n)}_{i+1}).$
Thus for any $i<n$
\begin{eqnarray}\label{3.10++}
&(1-\mu)\int_{\theta^{(n)}_i}^{\theta^{(n)}_{i+1}}S^W(t)d\gamma^{-}(t)-(1+\lambda)
\int_{\theta^{(n)}_i}^{\theta^{(n)}_{i+1}}S^W(t)d\gamma^{+}(t)\leq\\
&(1-\mu_n)S^W(\theta^{(n)}_{i+1})\int_{\theta^{(n)}_i}^{\theta^{(n)}_{i+1}}d\gamma^{-}(t)
-(1+\lambda_n)S^W(\theta^{(n)}_{i+1})\int_{\theta^{(n)}_i}^{\theta^{(n)}_{i+1}}d\gamma^{+}(t)\leq \nonumber\\
&S^W(\theta^{(n)}_{i+1})\big((1-\mu_n)(u_n(i+1)-u_n(i))^{-}-(1+\lambda_n)(u_n(i+1)-u_n(i))^{+}\big)\leq\nonumber\\
&S^W(\theta^{(n)}_{i+1})\big((1-\mu_n)(u_n(i)^{+}-u_n(i+1)^{+})-(1+\lambda_n)(u_n(i)^{-}-u_n(i+1)^{-})\big).\nonumber
\end{eqnarray}
Set $u_n(-1)=0$. From (\ref{2.6}) and (\ref{3.10++}) it follows that for any $k\leq n$
\begin{eqnarray}\label{3.10+++}
&V^{\pi_n}_{\lambda_n,\mu_n}(\theta^{(n)}_k)=x+(1-\mu_n)\big(u_n(k-1)^{+}S^{W}(\theta^{(n)}_k)+\\
&\sum_{i=0}^{k-2}(u_n(i+1)-
u_n(i))^{-} S^{W}(\theta^{(n)}_{i+1})\big)-(1+\lambda_n)
\big(u_n(k-1)^{-}S^{W}(\theta^{(n)}_k)+\nonumber\\
&\sum_{i=0}^{k-2} (u_n(i+1)-u_n(i))^{+}
S^{W}(\theta^{(n)}_{i+1})\big)\geq
x+(1-\mu)\big(u_n(k)^{+}S^{W}(\theta^{(n)}_k)+\nonumber\\
&\sum_{i=0}^{k-1}\int_{\theta^{(n)}_i}^{\theta^{(n)}_{i+1}}S^W(t)d\gamma^{-}(t)\big)
 -(1+\lambda)\big(u_n(k)^{-}S^{W}(\theta^{(n)}_k)+\nonumber\\
&+\sum_{i=0}^{k-1}\int_{\theta^{(n)}_i}^{\theta^{(n)}_{i+1}}S^W(t)d\gamma^{+}(t)\big)=V^\pi_{\lambda,\mu}(\theta^{(n)}_k\wedge T)\geq 0.\nonumber
\end{eqnarray}
Thus $\pi_n\in\mathcal{A}^{W,n}(x,\lambda_n,\mu_n)$.
From (\ref{3.2}) and (\ref{3.2++}) we obtain that there exists a stopping time
$\tau_n\in\mathcal{T}^{W,n}_{0,n}$ such that
\begin{equation}\label{3.11}
E^W[(Y^{W,n}(\tau_n T/n)-{V}^{\pi_n}_{\lambda_n,\mu_n}(\theta^{(n)}_{\tau_n}))^{+}]\geq
R_n(x,\lambda_n,\mu_n)-\epsilon.
\end{equation}
Clearly $R(\pi,\lambda,\mu)\geq E^W[(Y^W(\theta^{(n)}_{\tau_n}\wedge T)-V^\pi_{\lambda,\mu}(\theta^{(n)}_{\tau_n}\wedge T))^{+}]$,
and so from (\ref{3.10+++})
\begin{equation}\label{3.11+}
R(\pi,\lambda,\mu)\geq E^W[(Y^W(\theta^{(n)}_{\tau_n}\wedge T)-V^{\pi_n}_{\lambda_n,\mu_n}(\theta^{(n)}_{\tau_n}))^{+}].
\end{equation}
From
(\ref{3.10+}), (\ref{3.11}) and (\ref{3.11+}) it follows
\begin{equation}\label{3.12}
R_n(x,\lambda_n,\mu_n)\leq R(x,\lambda,\mu)+ 2\epsilon+ E^W|Y^{W,n}(\tau_n T/n)-Y^W(\theta^{(n)}_{\tau_n}\wedge T)|.
\end{equation}
By using the same arguments as in Lemma 3.1 we get $\lim_{n\rightarrow\infty}E^W|Y^{W,n}(\tau_n T/n)-Y^W(\theta^{(n)}_{\tau_n}\wedge T)|=0$
and we complete the proof.
\end{proof}
Observe that for any $n\in\mathbb{N}$, $\lambda'>0$, $0<\mu'<1$ and $x\geq 0$, the functions
$R_n(x,\lambda',\cdot), R_n(x,\cdot,\mu')$ are non decreasing. Thus from Lemma \ref{lem3.2} we obtain
that for any $\lambda>0$, $0<\mu<1$ and $0<\epsilon<1-\mu$
\begin{equation}\label{3.12+}
R(x,\lambda+\epsilon,\mu+\epsilon)\geq \lim \sup_{n\rightarrow\infty} R_n(x,\lambda,\mu).
\end{equation}
Define the function $\bar{R}:\mathbb{R}_{+}\times (0,\infty)\times (0,\infty)\rightarrow \mathbb{R}_{+}$ by
\begin{equation}\label{3.12++}
\bar{R}(x,\lambda,\mu)=\lim_{\lambda'\downarrow\lambda}\lim_{\mu'\downarrow\mu}R(x,\lambda',\mu').
\end{equation}
The limit above is exists since the functions $R_n(x,\lambda',\cdot), R_n(x,\cdot,\mu')$ are non decreasing.
From (\ref{3.12+})
\begin{equation}\label{3.12+++}
\bar{R}(x,\lambda,\mu)\geq \lim \sup_{n\rightarrow\infty} R_n(x,\lambda,\mu).
\end{equation}
Next, fix $\lambda,\mu$ and let $\pi_n=\pi_n(x,\lambda,\mu)$ be the optimal portfolio which is given by
(\ref{2+.24}).
From Lemma \ref{lem3.1} and (\ref{3.12+++}) we obtain
\begin{eqnarray}\label{3.13--}
&R(x,\lambda,\mu)\leq \lim \inf_{n\rightarrow\infty} R(\psi_n(\pi_n),\lambda,\mu)\leq
\lim \inf_{n\rightarrow\infty} R_n(x,\lambda,\mu)\leq \\
&\bar{R}(x,\lambda,\mu) \ \
\mbox{and} \ \ R(x,\lambda,\mu)\leq \lim \sup_{n\rightarrow\infty} R(\psi_n(\pi_n),\lambda,\mu)\leq\nonumber\\
&\lim \sup_{n\rightarrow\infty} R_n(x,\lambda,\mu)\leq \bar{R}(x,\lambda,\mu). \nonumber
\end{eqnarray}
Thus in order to complete the proof of Theorems \ref{thm2.1}--\ref{thm2.2}, it remains to establish
the following stability result.
\begin{lem}\label{lem3.3}
For any $\lambda>0$, $0<\mu<1$ and $x\in\mathbb{R}_{+}$
\begin{equation}\label{3.13-}
\bar{R}(x,\lambda,\mu)=R(x,\lambda,\mu).
\end{equation}
\end{lem}
\begin{proof}
The inequality $\bar{R}(x,\lambda,\mu)\geq R(x,\lambda,\mu)$, is trivial. Thus it is sufficient to show that
$\bar{R}(x,\lambda,\mu)\leq R(x,\lambda,\mu)$. Fix $\lambda,\mu,x$ and
choose $\epsilon>0$. For $x=0$ the statement is trivial since $R(0,\cdot,\cdot)\equiv \sup_{\tau\in\mathcal{T}^W_{[0,T]}}E^W Y^W(\tau)$.
Assume that $x>0$.
There exists $\pi=(x,\{\gamma(t)\}_{t=0}^T)\in\mathcal{A}^W(x,\lambda,\mu)$
such that
\begin{equation}\label{3.13}
R(\pi,\lambda,\mu)<R(x,\lambda,\mu)+\epsilon.
\end{equation}
Set
\begin{equation}\label{3.13+}
\lambda^{(n)}_q=\lambda+\frac{(1-q)x}{n} \ \mbox{and} \ \mu^{(n)}_q=\mu+\frac{(1-q)x}{n}, \ n\in\mathbb{N}, \ 0<q<1.
\end{equation}
We assume that $n$ is sufficiently large such that $\mu^{(n)}_q<1$ for any $0<q<1$.
Introduce the stopping times
\begin{equation}\label{3.13++}
\tau_n=T\wedge\inf\bigg\{t\big|\int_{0}^tS^W(u)|d\gamma|(u)+|\gamma(t)|S^W(t)\geq n\bigg\}, \ \ n\in\mathbb{N}.
\end{equation}
The stochastic process $\big\{\int_{0}^tS^W(u)|d\gamma|(u)+|\gamma(t)|S^W(t)\big\}_{t=0}^T$
is left continuous, and so for any $t\leq T$,
\begin{equation}\label{3.13+++}
\int_{0}^{t\wedge\tau_n}S^W(u)|d\gamma|(u)+|\gamma(t\wedge\tau_n)|S^W(t\wedge\tau_n)\leq n.
\end{equation}
Notice that
\begin{equation}\label{3.14--}
\lim_{n\rightarrow\infty}\tau_n=T \ \ \mbox{a.s.}
\end{equation}
From (\ref{3.13+}) and (\ref{3.13+++}) it follows that for any $0\leq t\leq T$
\begin{eqnarray}\label{3.14-}
&(\mu^{(n)}_q-\mu)\big( \int_{0}^{t\wedge\tau_n}S^W(u)d\gamma^{-}(u)+\gamma(t\wedge\tau_n)^{+}S^W(t\wedge\tau_n)\big)+
\\
&(\lambda^{(n)}_q-\lambda)\big(\int_{0}^{t\wedge\tau_n}S^W(u)d\gamma^{+}(u)+\gamma(t\wedge\tau_n)^{-}S^W(t\wedge\tau_n)\big) \leq (1-q)x.
\nonumber
\end{eqnarray}
For any $n\in\mathbb{N}$ and $0<q<1$, $\{q\gamma(t)\mathbb{I}_{t\leq\tau_n}\}_{t=0}^T$ is
an adapted process of bounded variation with left continuous paths. Consider
the portfolio $\pi^{(n)}_q=(x,\{q\gamma(t)\mathbb{I}_{t\leq\tau_n}\}_{t=0}^T)$.
From (\ref{2.6}) and (\ref{3.14-}) we obtain
\begin{eqnarray}\label{3.15}
&V^{\pi^{(n)}_q}_{\lambda^{(n)}_q,\mu^{(n)}_q}(t)=V^{\pi^{(n)}_q}_{\lambda^{(n)}_q,\mu^{(n)}_q}(t\wedge\tau_n)=
qx+(1-q)x+\\
&q(1-\mu^{(n)}_q)\big(\int_{0}^{t\wedge\tau_n} S^W(u)d\gamma^{-}(u)+\gamma(t\wedge\tau_n)^{+}S^W(t\wedge\tau_n)\big)-\nonumber\\
&q(1+\lambda^{(n)}_q)
\big(\int_{0}^{t\wedge\tau_n} S^W(u)d\gamma^{+}(u)+\gamma(t\wedge\tau_n)^{-}S^W(t\wedge\tau_n)\big)\geq qx+\nonumber\\
&q(1-\mu)
\big(\int_{0}^{t\wedge\tau_n}S^W(u)d\gamma^{-}(u)+\gamma(t\wedge\tau_n)^{+}S^W(t\wedge\tau_n)\big)-q(1+\lambda)\times\nonumber\\
&\big(\int_{0}^{t\wedge\tau_n} S^W(u)d\gamma^{+}(u)+\gamma(t\wedge\tau_n)^{-}S^W(t\wedge\tau_n)\big)=q V^\pi_{\lambda,\mu}(t\wedge\tau_n)\geq 0.\nonumber
\end{eqnarray}
We conclude that $\pi^{(n)}_q\in\mathcal{A}^W(x,\lambda^{(n)}_q,\mu^{(n)}_q)$. From (\ref{3.13})
and (\ref{3.15})
\begin{eqnarray}\label{3.16}
&R(x,\lambda^{(n)}_q,\mu^{(n)}_q)\leq \sup_{\tau\in\mathcal{T}^W_{[0,T]}}
E^W[(Y^W(\tau)-q V^\pi_{\lambda,\mu}(\tau\wedge\tau_n))^{+}]\leq \\
&q \sup_{\tau\in\mathcal{T}^W_{[0,T]}} E^W[(Y^W(\tau)- V^\pi_{\lambda,\mu}(\tau\wedge\tau_n))^{+}]+\nonumber\\
&(1-q)E^W\sup_{0\leq t\leq T}Y^W(t)
\leq \epsilon+R(x,\lambda,\mu)+\nonumber\\
&E^W\sup_{0\leq t\leq T}|Y^W(t)-Y^W(t\wedge\tau_n)|
+(1-q)E^W\sup_{0\leq t\leq T}Y^W(t).\nonumber
\end{eqnarray}
For any $n$, $\bar{R}(x,\lambda,\mu)= \lim_{q\uparrow 1} R(x,\lambda^{(n)}_q,\mu^{(n)}_q)$, and so from
(\ref{3.16}) it follows that for any $n$,
$\bar{R}(x,\lambda,\mu)\leq\epsilon+R(x,\lambda,\mu)+E^W\sup_{0\leq t\leq T}|Y^W(t)-Y^W(t\wedge\tau_n)|.$
Thus
\begin{equation}\label{3.17}
\bar{R}(x,\lambda,\mu)\leq\epsilon+R(x,\lambda,\mu)+\lim \inf_{n\rightarrow\infty}E^W\sup_{0\leq t\leq T}|Y^W(t)-Y^W(t\wedge\tau_n)|.
\end{equation}
From (\ref{3.14--}) we obtain that $\lim_{n\rightarrow\infty}\sup_{0\leq t\leq T}|Y^W(t)-Y^W(t\wedge\tau_n)|=0$ a.s.
Thus $\lim_{n\rightarrow\infty}E^W\sup_{0\leq t\leq T}|Y^W(t)-Y^W(t\wedge\tau_n)|=0$,
this together with (\ref{3.17}) completes the proof.
\end{proof}
\begin{rem}\label{rem3.1}
Consider the BS model in the absence of transaction costs (complete market).
In this case a self financing strategy $\pi$ with an initial capital $x$ is a
pair $(x,\{\gamma(t)\}_{t=0}^T)$ such that the process
$\{\gamma(t)\}_{t=0}^T$ is  progressively measurable with respect to
the filtration $\mathcal{F}^{W}_t$, $t\geq{0}$
and satisfy
\begin{equation}\label{3.18}
\int_{0}^{T}\big(\gamma(t){S}^{W}(t)\big)^2dt<\infty \ \ \ \ \mbox{a.s.}
\end{equation}
The portfolio value $V^\pi(t)$ for a strategy $\pi=(x,\{\gamma(t)\}_{t=0}^T)$ at time $t\in{[0,T]}$
 is given by
\begin{equation}\label{3.19}
{V}^{\pi}(t)=x+\int_{0}^{t}
\gamma(u)d{{S}^W(u)}.
\end{equation}
 A self financing
strategy $\pi$ is called \textit{admissible} if ${V}^{\pi}(t)\geq{0}$
for all $t\in{[0,{T}]}$ and the set of such strategies with an
initial capital $x$ will be denoted by $\mathcal{A}^W(x)$. The shortfall risk is defined by
\begin{equation}\label{3.20}
R(\pi)=\sup_{\tau\in\mathcal{T}^W_{[0,T]}}E^W[(Y^W(\tau)-{V}^{\pi}(\tau))^+] \ \ \mbox{and} \ \
R(x)=\inf_{\pi\in\mathcal{A}^W(x)}R(\pi).
\end{equation}
Let $\mathcal{A}^W(x,0,0)\subset\mathcal{A}^W(x)$ be the set of all
portfolios $(x,\{\gamma(t)\}_{t=0}^T)$ such that $\{\gamma(t)\}_{t=0}^T$
is an adapted process of bounded variation with left continuous paths and $\gamma(0)=0$.
Note that for any $(x,\{\gamma(t)\}_{t=0}^T)\in\mathcal{A}^W(x,0,0)$ the portfolio values which are given by
(\ref{2.6}) with $\lambda=\mu=0$ and (\ref{3.19}) are coincide. From Dolinsky and Kifer (2008) (Theorem  2.2) it
follows that for any initial capital $x\in\mathbb{R}_{+}$ and $\epsilon>0$ there exists $n\in\mathbb{N}$ and a portfolio $\pi=(x,\{\gamma(t)\}_{t=0}^T)$
of the form
\begin{equation}\label{3.21}
\gamma(t)=\sum_{i=0}^{n-1} \mathbb{I}_{\theta^{(n)}_i< t\leq\theta^{(n)}_{i+1}}u_{i+1}
\end{equation}
where for any $1\leq i\leq n$, $u_i$ is a random variable $\mathcal{F}^W_{\theta^{(n)}_{i-1}}$
measurable, such that $R(\pi)<R(x)+\epsilon$.
Thus $R(x)=\inf_{\pi\in\mathcal{A}^W(x,0,0)}R(\pi)$, and so by following the steps of the proof of Lemma \ref{lem3.3}
we get that for any initial capital $x\geq 0$
\begin{equation}\label{3.22}
R(x)=\lim_{\lambda\downarrow 0}\lim_{\mu\downarrow 0}R(x,\lambda,\mu).
\end{equation}
Consider an American call option $Y^W(t)=(S^W(t)-K e^{-rt})^{+}$, $t\leq T$
with parameters $K,r>0$.
Clearly $V^{*}=\tilde{E}^W Y^W(T)$ is the price of the above call option in the complete BS model. From (\ref{3.22})
it follows that $\lim_{\lambda\downarrow 0}\lim_{\mu\downarrow 0}R(V^{*},\lambda,\mu)=0$. In particular we obtain
that in the presence of transaction costs, for an initial capital $x=V^{*}$ and for sufficiently small
$\lambda,\mu>0$ the buy and hold strategies are not optimal (unlike
for the superhedging case) for the shortfall risk measure.
\end{rem}
\section{Proof of Theorem 2.1}\label{sec4}
In this section we assume that the parameters
$x,\lambda,\mu$ are fixed.
Let $I\subset [0,T]$ be a dense set in $[0,T]$ and let $\mathcal{T}^W_I\subset\mathcal{T}^W_{[0,T]}$
be the set of all stopping times with a finite number of values which belongs to $I$.
\begin{lem}\label{lem4.1}
For any $\pi\in\mathcal{A}^W(x,\lambda,\mu)$,
\begin{equation}\label{4.1}
R(\pi,\lambda,\mu)=\sup_{\tau\in \mathcal{T}^W_I}E^W[(Y^W({\tau})-V^\pi_{\lambda,\mu}({\tau}))^{+}].
\end{equation}
\end{lem}
\begin{proof}
Clearly $R(\pi,\lambda,\mu)\geq\sup_{\tau\in \mathcal{T}^W_I}E^W[(Y^W({\tau})-V^\pi_{\lambda,\mu}({\tau}))^{+}]$.
Thus it is sufficient to show that $R(\pi,\lambda,\mu)\leq\sup_{\tau\in \mathcal{T}^W_I}E^W[(Y^W({\tau})-V^\pi_{\lambda,\mu}({\tau}))^{+}]$.
Choose $\epsilon>0$. There exists $\tau\in\mathcal{T}^W_{[0,T]}$ such that
\begin{equation}\label{4.2}
R(\pi,\lambda,\mu)< E^W[(Y^W({\tau})-V^\pi_{\lambda,\mu}({\tau}))^{+}]+\epsilon.
\end{equation}
For any $n$ there exists a finite set $I_n\subset I$ for which $\bigcup_{z\in I_n} (z-\frac{1}{n},z+\frac{1}{n})\supseteq [0,T]$.
Let $a_n$ be the maximal element of $I_n$. Define $\tau_n=\min\{t\in I_n|t\geq\tau\}\mathbb{I}_{\tau_n\leq a_n}+a_n \mathbb{I}_{\tau_n>a_n}$.
Clearly, $\tau_n\leq a_n$ a.s. and for $t\in I_n\setminus\{a_n\}$ we have $\{\tau_n\leq t\}=\{\tau\leq t\}\in \mathcal{F}^W_t$.
Thus $\tau_n\in \mathcal{T}^W_I$. Furthermore, $|\tau_n-\tau|\leq \frac{2}{n}$ and so $\tau_n\rightarrow\tau$ a.s.
From (\ref{2.6}) it follows that the stochastic process $\{V^\pi_{\lambda,\mu}(t)\}_{t=0}^T$
is left continuous with right hand limits and
has only negative jumps (in discontinuity points). Thus $V^\pi_{\lambda,\mu}(\tau)\geq \lim \sup_{n\rightarrow\infty} V^\pi_{\lambda,\mu}(\tau_n)$ a.s.
By using (\ref{4.2}) and Fatou's lemma we obtain
\begin{eqnarray}\label{4.3}
&R(\pi,\lambda,\mu)< \epsilon+E^W[\lim \inf_{n\rightarrow\infty}(Y^W(\tau_n)-V^\pi_{\lambda,\mu}({\tau_n}))^{+}]\leq\epsilon+\\
&\lim \inf_{n\rightarrow\infty}E^W[(Y^W(\tau_n)-V^\pi_{\lambda,\mu}(\tau_n))^{+}]\leq
\epsilon+\sup_{\tau\in \mathcal{T}^W_I}E[(Y^W({\tau})-V^\pi_{\lambda,\mu}({\tau}))^{+}]\nonumber
\end{eqnarray}
and the result follows by letting $\epsilon\downarrow{0}$.
\end{proof}
Next, let
$\{\pi_n=(x,\gamma_n)\}_{n=1}^\infty\subset\mathcal{A}^W(x,\lambda,\mu)$ be
a sequence such that
\begin{equation}\label{4.4}
\lim_{n\rightarrow\infty} R(\pi_n,\lambda,\mu)=R(x,\lambda,\mu).
\end{equation}
From the integration by part formula we
get that for any $n\in\mathbb{N}$ and $t\in [0,T]$
$\gamma_n(t)S^W(t)=\int_{0}^tS^W(u)d\gamma_n(u)+\int_{0}^t\gamma_n(u)dS^W(u).$
This together with (\ref{2.6})--(\ref{2.7}) yields
\begin{equation}\label{4.5}
\min(\lambda,\mu)\int_{0}^t S^W(u)|d\gamma_n|(u)\leq x+\int_{0}^t\gamma_n(u)dS^W(u), \ \ t\in[0,T], \ \ n\in\mathbb{N}.
\end{equation}
From (\ref{4.5}) it follows that any $n$, the local martingale (with respect to the probability measure
$\tilde{P}^W$) $\{\int_{0}^t\gamma_n(u)dS^W(u)\}_{t=0}^T$ is bounded from below, and so it is
a supermartingale.
Thus from (\ref{4.5}), $\tilde{E}^W\int_{0}^T S^W(u)|d\gamma_n|(u)\leq \frac{x}{\min(\lambda,\mu)}$, $n\in\mathbb{N}$.
From Markov's inequality we get that the set
$conv\big\{\int_{0}^T S^W(u)|d\gamma_n|(u)\big\}_{n=1}^\infty$ is bounded in $L^{0}(\tilde{P}^W)$.
This together with Lemma 3.1 in Guasoni (2002B) yields that the set
$conv\big\{\int_{0}^T |d\gamma_n|(u)\big\}_{n=1}^\infty$ is also bounded in $L^{0}(\tilde{P}^W)$.
From Lemma 3.4 in
Guasoni (2002B) there is a sequence $\eta_n\in conv(\gamma_n,\gamma_{n+1},...)$
such that $\eta_n$ converges a.s. in $dt d\tilde{P}^W$ to a finite variation process.
In fact, from the proof of this lemma, we get a stronger result. We obtain
that there exists a non decreasing, left continuous adapted processes $\{\alpha(t)\}_{t=0}^T$ and $\{\beta(t)\}_{t=0}^T$ 
with $\alpha(0)=\beta(0)=0,$
such that
\begin{equation}\label{4.7}
\lim_{n\rightarrow\infty}\eta^{+}_n= \alpha \ \ \mbox{and} \ \ \lim_{n\rightarrow\infty}\eta^{-}_n=\beta, \ \ \ \mbox{a.s} \ \mbox{in} \ dtd\tilde{P}^W
\end{equation}
where
\begin{equation}\label{4.7+}
\eta_n^{+}(t)=\frac{\eta_n(t)+ \int_{0}^t |d\eta_n|(s)}{2} \ \ \mbox{and} \ \
\eta_n^{-}(t)=\frac {\int_{0}^t |d\eta|(s)-\eta_n(t)}{2}, \ \ t\in[0,T], \ \ n\in\mathbb{N}.
\end{equation}
In particular, there exists a countable dense set $0\in I\subset [0,T]$ such that
\begin{equation}\label{4.8}
P^W\{\lim_{n\rightarrow\infty}\eta^{+}_n(t)= \alpha(t), \ \ \forall{t}\in I\}=1 \ 
\mbox{and} \
P^W\{\lim_{n\rightarrow\infty}\eta^{-}_n(t)= \beta(t), \ \ \forall{t}\in I\}=1. 
\end{equation}
Define $\gamma=\alpha-\beta$. Clearly, $\gamma$
is an adapted process of bounded variation with left continuous paths and $\gamma(0)=0$.
Finally, we prove that $\pi:=(x,\gamma)$ is an optimal portfolio, i.e., 
$\pi\in\mathcal{A}^W(x,\lambda,\mu)$ and $R(\pi,\lambda,\mu)=R(x,\lambda,\mu)$. 
Clearly for any $n\in\mathbb{N}$,
the wealth process of the portfolio $\tilde\pi_n:=(x,\eta_n)$ is satisfying
$V^{\tilde\pi_n}_{\lambda,\mu}\in conv\{V^{\pi_n}_{\lambda,\mu}, V^{\pi_{n+1}}_{\lambda,\mu},...\}$,
and so $\tilde\pi_n\in\mathcal{A}^W(x,\lambda,\mu)$.
The shortfall risk measure $R(\cdot,\lambda,\mu)$ is a a convex
functional of the wealth process $V^{\cdot}_{\lambda,\mu}$. Thus,
\begin{equation}\label{4.9}
R(\tilde\pi_n,\lambda,\mu)\leq \sup_{k\geq n} R(\pi_k,\lambda,\mu).
\end{equation}
From (\ref{4.4}) and (\ref{4.9}),
\begin{equation}\label{4.10}
\lim_{n\rightarrow\infty} R(\tilde\pi_n,\lambda,\mu)=R(x,\lambda,\mu).
\end{equation}
From (\ref{4.8}) and Theorem 12.16 in Protter and Morrey (1991),
\begin{eqnarray}\label{4.11}
&\int_{0}^t S^W(u)d\alpha(u)=\lim_{n\rightarrow\infty}S^W(u)d\eta^{+}_n(u) \ \ \mbox{and} \\
&\int_{0}^t S^W(u)d\beta(u)=\lim_{n\rightarrow\infty}S^W(u)d\eta^{-}_n(u), \ \  \mbox{a.s.} \ \ \forall{t\in I}.\nonumber
\end{eqnarray}
Thus
\begin{eqnarray}\label{4.11+}
&\int_{0}^t S^W(u)d\gamma(u)=\lim_{n\rightarrow\infty}S^W(u)d\eta_n(u) \ \ \mbox{and} \ \ 
\int_{0}^t S^W(u)|d\gamma|(u)\leq\\
&\int_{0}^t S^W(u)d\alpha(u)+\int_{0}^t S^W(u)d\beta(u)=\lim_{n\rightarrow\infty}
\int_{0}^t S^W(u)|d\eta_n|(u),  \ \  \mbox{a.s.} \ \ \forall{t\in I}.\nonumber
\end{eqnarray}
This together with (\ref{2.5})--(\ref{2.6}) gives
\begin{equation}\label{4.13}
V^{\pi}_{\lambda,\mu}(t)\geq\lim _{n\rightarrow\infty} V^{\tilde\pi_n}_{\lambda,\mu}(t)\geq 0, \ \ \forall t\in I.
\end{equation}
Thus $\pi\in\mathcal{A}^W(x,\lambda,\mu)$. By combining  Fatou's lemma together 
with Lemma \ref{lem4.1}, (\ref{4.10}) and (\ref{4.13})
we obtain
\begin{eqnarray*}
&R(\pi,\lambda,\mu)=\sup_{\tau\in \mathcal{T}^W_I}E^W[(Y^W({\tau})-V^{\pi}_{\lambda,\mu}({\tau}))^{+}]\leq\\
&\sup_{\tau\in \mathcal{T}^W_I}E^W[\lim_{n\rightarrow\infty}(Y^W({\tau})-V^{\tilde\pi_n}_{\lambda,\mu}({\tau}))^{+}]\leq\\
&\sup_{\tau\in \mathcal{T}^W_I}\lim \inf_{n\rightarrow\infty}E^W[(Y^W({\tau})-V^{\tilde\pi_n}_{\lambda,\mu}({\tau}))^{+}]\leq\\
&\lim_{n\rightarrow\infty} R(\tilde\pi_n,\lambda,\mu)=R(x,\lambda,\mu).
\end{eqnarray*}
Thus $R(\pi,\lambda,\mu)=R(x,\lambda,\mu)$ and the proof is completed.
\qed
\\
\\

\end{document}